\documentclass[a4paper,pra,twocolumn,nofootinbib]{revtex4-2} 
\usepackage[T1]{fontenc}
\usepackage{times}

\usepackage{amsmath}
\usepackage{amsthm}
\usepackage{amssymb}
\usepackage[utf8]{inputenc}
\usepackage{color}
\usepackage{hyperref}
\usepackage{tikz}

\newtheorem{theorem}{Theorem}
\newtheorem{lemma}[theorem]{Lemma}

\def\Real{{\mathbb R}}
\def\Cmpx{{\mathbb C}}

\def\cnj#1{{\overline{#1}}}

\def\Re{{\textup{Re}}}
\def\Im{{\textup{Im}}}

\def\innerprod(#1,#2){{\left<#1\,,\,#2\right>}}
\def\Set#1{{\left\{#1\right\}}}
\def\qquadtext#1{\qquad\textup{#1}\qquad}
\def\qquadand{\qquadtext{and}}
\def\quadtext#1{\quad\textup{#1}\quad}
\def\quadand{\quadtext{and}}

\def\sind{{s}}

\def\VD{{\boldsymbol D}}
\def\VP{{\boldsymbol P}}
\def\VH{{\boldsymbol H}}
\def\VE{{\boldsymbol E}}
\def\VB{{\boldsymbol B}}
\def\Vx{{\boldsymbol x}}
\def\Vk{{\boldsymbol k}}
\def\Vxhat{\hat{\boldsymbol x}}
\def\Vyhat{\hat{\boldsymbol y}}
\def\Vzhat{\hat{\boldsymbol z}}

\def\Jcurrent{J}
\def\VJ{{\boldsymbol{\Jcurrent}}}
\def\VJtotal{\VJ_{\!\textup{total}}}
\def\VJusual{\VJ_{\!\textup{reg}}}
\def\VJsurface{{\boldsymbol{\mathcal{\Jcurrent}}_{\!\!\textup{surf}}}}

\def\Xafter{{\textup{c}}}
\def\Xbound{{\textup{b}}}
\def\Xbfore{{\textup{a}}}

\def\aafactor{g}
\def\Amp{{\aafactor^{\Xbfore}_{\Xafter}}}
\def\Amm{{\aafactor^{\Xbfore}_{\Xbfore}}}
\def\App{{\aafactor^{\Xafter}_{\Xafter}}}
\def\Apm{{\aafactor^{\Xafter}_{\Xbfore}}}

\def\Amp{{\aafactor^{-}_{+}}}
\def\Amm{{\aafactor^{-}_{-}}}
\def\App{{\aafactor^{+}_{+}}}
\def\Apm{{\aafactor^{+}_{-}}}

\def\tinyzero{{\scalebox{0.6}{0}}}

\def\co{{c_{\tinyzero}}}
\def\cosq{{c_{\tinyzero}^2}}

\def\tzero{{t_{\Xbound}}}

\def\wafter{\omega_{\Xafter}}
\def\wbfore{\omega_{\Xbfore}}
\def\kafter{k_{\Xafter}}
\def\kbfore{k_{\Xbfore}}

\def\nafter{n_{\Xafter}}
\def\nbfore{n_{\Xbfore}}

\def\cafter{c_{\Xafter}}

\def\cc{{c}}
\def\cI{\cc_{\textup{I}}}
\def\cR{\cc_{\textup{R}}}

\def\CC{{C}}
\def\CI{\CC_{\textup{I}}}
\def\CR{\CC_{\textup{R}}}

\def\FTepsilon{{\widetilde{\epsilon}}}

\def\FTn{{\tilde{\nind}}}

\def\nind{{n}}
\def\nconst{{N}}

\def\epsilonII{{\epsilon_{\Xafter}}}

\newcommand{\XDOI}[1]{\href{http://dx.doi.org/#1}{doi:#1}}
\newcommand{\XARXIV}[1]{\href{http://arxiv.org/abs/#1}{arXiv:#1}}
\newcommand{\XWEB}[1]{\href{#1}{#1}}

\def\pCurrent{J}
\def\Kfield{K}

\def\pEfield{E}

\def\pPermittivity{\varepsilon}

\def\pSlave{\Gamma}

\def\pPderivT{\partial_t}

\begin{document}

\title{Temporal boundaries in electromagnetic materials}

\author{Jonathan Gratus$^{1,2}$} 
\homepage[]{https://orcid.org/0000-0003-1597-6084}

\author{Rebecca Seviour$^{3}$}
\homepage[]{https://orcid.org/0000-0001-8728-1463}

\author{Paul Kinsler$^{1,2,4}$}
\homepage[]{https://orcid.org/0000-0001-5744-8146}

\author{Dino A. Jaroszynski$^{5}$}
\homepage[]{https://orcid.org/0000-0002-3006-5492}

\affiliation{$^1$ 
  Department of Physics,
  Lancaster University,
  Lancaster LA1 4YB,
  United Kingdom,
}
\affiliation{$^2$ 
The Cockcroft Institute,
Sci-Tech Daresbury,
Daresbury WA4 4AD,
United Kingdom,
}

\affiliation{$^3$
 University of Huddersfield,
 Huddersfield HD1 1JB,
 United Kingdom,
}

\affiliation{$^4$
  Department of Physics,
  Imperial College London,
  Prince Consort Road,
  London SW7 2AZ,
  United Kingdom.
}

\affiliation{$^5$
Department of Physics,
SUPA and University of Strathclyde,
Glasgow G4 0NG,
United Kingdom.
}

\begin{abstract}
Temporally modulated optical media are important
 in both abstract and applied applications, 
 such as 
 spacetime transformation optics, 
 relativistic laser-plasma interactions,
 and dynamic metamaterials. 
Here we investigate the behaviour of temporal boundaries,
 and show that traditional approaches
 that assume constant dielectric properties,
 with loss incorporated as an imaginary part, 
 necessarily lead to unphysical solutions.
Further, 
 although physically reasonable predictions can be recovered
 with a narrowband approximation, 
 we show that
 appropriate models should
 use materials with a temporal response and dispersive behaviour.
\end{abstract}
\maketitle


\section{Introduction}
\label{ch_Intro}

Mobile technologies are considered
 the biggest technology platform in history, 
 with transformative advances occurring across all society.
These are having a profound impact in diverse areas
 including health care, 
 education, 
 and industry
 \cite{NICR-2016-Connected}. 
Key to developing mobile technologies is the ability 
 to predict electromagnetic (EM) wave propagation though systems
 with differing permittivities, 
 and in particular predicting losses.
In both physics \cite{YuCardona-FSPMP}
 and 
 engineering \cite{RWV-FieldsCommElectr},
 EM loss is often expressed using the electric loss tangent
 (``$\tan \delta$''),
 the ratio of the imaginary to the real part of the permittivity.
This constant permittivity model
 is widely used in condensed matter physics \cite{YuCardona-FSPMP}, 
 and is critical 
 to the design of technologies diverse as
 mobile phones, 
 imaging systems, 
 consumer electronics, 
 radar, 
 accelerators,
 sensors, 
 and even microwave therapy
 \cite{Alabaster-2004phd}.
In such situations, 
 microscopic models
 involving e.g. atomic structure or quantum mechanical effects 
 typically provide no significant advantage.

As new materials are developed for use in
 novel devices, 
 it becomes vital that EM loss and propagation are correctly predicted.  
A very exciting new concept is time dependent media,
 where abrupt changes in permittivity can 
 create
 a \emph{temporal boundary}.
Boundaries play a key role in many physical models;
 they provide initial and final states in dynamical systems,
 constrain analytic solutions in confined systems,
 and represent transitions between different modes of operation.
These concepts 
 date back to the 1950s, 
 when 
 Morgenthaler \cite{Morgenthaler-1958irs}
 showed that a temporal change of the permittivity
 produces both forward and backward propagating waves;
 a result echoed in 
 directional formulations
 for wave propagation
 \cite{Kolesik-M-2004pre,Mizuta-NOY-2005pra,Genty-KKD-2007oe,Kinsler-2018jo-d2owe,Kinsler-2018jpco-fbacou}.
Temporal boundaries \cite{Morgenthaler-1958irs,Xiao-MA-2014ol,Bakunov-M-2014ol,Tan-LZ-2020ol}
  act as time-reversing mirrors
 in acoustics \cite{deRosny-F-2002prl};
 in EM an instantaneous time mirror
 with a sign-change in permittivity
 has been predicted \cite{Kiasat-PEN-2018cleo}
 to cause {field} amplification.
Other examples include dynamically configurable systems
  \cite{Turpin-BMWW-2014ijap,McCall-etal-2018jo-roadmapto},
 spacetime transformation devices
 \cite{McCall-FKB-2011jo,Gratus-KMT-2016njp-stdisp,Kinsler-M-2014adp-scast,Kinsler-M-2014pra},
 time crystals \cite{Sacha-Z-2018rpp,Shapere-W-2012prl-classical,Wilczek-2012prl-quantum},
 ``field patterns'' \cite{Milton-M-2017rspa},
 and in laser driven plasma
 where relativistic changes 
 can lead to a spatial and temporally varying permittivity.

Here we prove that 
 modelling loss
 using a constant complex permittivity and permeability
 \cite{Pendry-2000prl,Dolin-1961ivuzr,Kiasat-PEN-2018cleo,Kinsler-M-2008motl}
 is physically incompatible with a temporal boundary.
Such models can lead to unphysical post-boundary solutions
 that grow exponentially, 
 despite being applied to passive and lossy materials; 
 or fields may become complex-valued despite being real-valued before.
This important point,
 which provides an unambiguous warning to non-specialists, 
 is not addressed in other recent work 
 on temporal boundaries, 
 which focus on either reflection and refraction 
 inside a medium with parabolic dispersion
 \cite{Zhang-DA-2021josab}, 
 appropriately generalised Kramers Kronig relations \cite{Solis-E-2021prb},
 or the complicated effects resulting
 from a Lorentzian response model \cite{Solis-KE-2021arxiv}.
Our dynamic material model, 
 in contrast to more complicated ones
 \cite{Solis-E-2021prb,Solis-KE-2021arxiv}
 is explicitly designed to provide a minimal example with simple behaviour
 which clearly reveals the basic physical principles
 relating to the treatment of temporal boundaries:
 the necessity of considering the dynamics of the bound current, 
 the requirement for material-property boundary conditions, 
 and the resulting secondary implications for frequency-domain properties
 such as the dispersion relations.

\begin{figure}[t]
\centering
\resizebox{0.95\columnwidth}{!}{\includegraphics{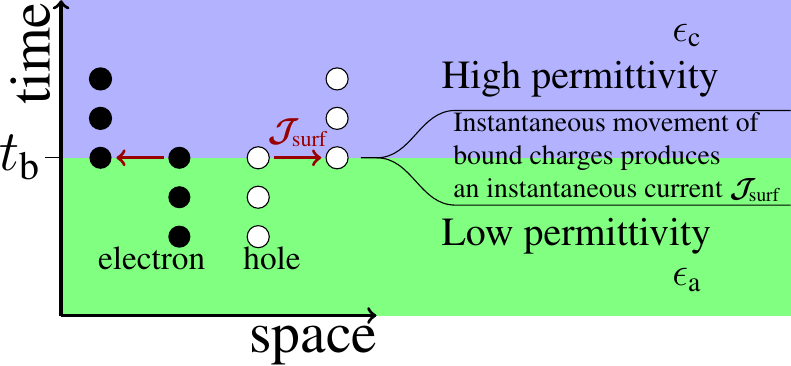}}
\caption{A temporal boundary or transition
 based on {fundamental} TBCs.
Here, 
 a sudden increase of permittivity at $t={{\protect{\tzero}}}$, 
 represented as the movement of bound charges,
 necessarily 
 generates a microscopic ``time-surface'' current $\VJsurface$.
}
\label{fig_TempSurf}
\end{figure}

In this article the term ``boundary'' refers
 to an interface between two regions
 with different constitutive relations (CRs). 
In electromagnetism CRs are most simply given by 
 a permittivity $\epsilon$ and permeability $\mu$,
 although more complicated CRs are also allowed.
In contrast to 
 spatial boundary conditions 
 that describe an interface between two static media, 
 here we consider a temporal boundary at $\tzero$,
 where the medium has one set of CR properties
 before the transition ($t < \tzero$),
 and different CRs afterwards ($t > \tzero$). 
In our idealised transition,
 the CR for the entire region changes instantly;
 although gradual transitions are also possible
 \cite{Chen-GGL-2009mas,Chen-L-2011mas}.
The unphysical consequences arise irrespective 
 of which of the two possible types of 
 temporal boundary condition
 (TBC) we consider.
The first
 ``natural'' TBC \cite{Cheng-C-2005eabe}
 is derived from Maxwell's equations on the assumption
 that all the \emph{currents} are finite,
 which leads to the continuity of $\VD$ and $\VB$. 
The second {``fundamental''} TBC
 treats only $\VE$ and $\VB$ as
 physical EM fields,
 with $\VD$ and $\VH$ as derived fields
 acting as a gauge for the current
 \cite{Gratus-KM-2019ejp-dhfield,Gratus-KM-2019foop-nocharge,Gratus-MK-2020pra-area51}; 
 here $\VE$ and $\VB$ are also continuous, 
 but a \emph{temporal-surface} dipole current appears
 at the transition,
 as shown in Fig. \ref{fig_TempSurf}.

After demonstrating and defining the problem,
 we present two methods for obtaining physically meaningful solutions. 
The first applies the
 constant complex CR model 
 whilst using a narrow band approximation (NBA).
This leads to a solution based on
 complex conjugate pairs of frequencies and refractive indices.
Although partially successful, 
 it
 merely hides the fact that to model lossy media correctly
 when there is a temporal boundary,
 one requires a time-dependent material response,
 and therefore dispersive CR,
 where $\epsilon$ and/or $\mu$
 depend on frequency. 
Besides,
 just as a dynamic medium model requires its own fields 
 to represent its state, 
 in the frequency domain we see that a dispersive medium
 generates one or more additional modes; 
 and this necessary information 
 cannot be included in the standard constant complex CR model.
Thus in either time or frequency, 
 additional boundary conditions (ABCs)
 must be specified.

Note that 
 proofs and additional discussion are presented
 in the Supplementary Material following the References.

\section{Linear Media}
In linear media
 any EM field can be represented in the Fourier domain
 by a sum or integral of terms of the form
 $\exp({-i\omega t + i \Vk\cdot\Vx})$,
 where $\omega\in\Cmpx$ is a complex frequency,
 and
 $\Vk\in\Cmpx^3$ a complex wavevector.
Since the source free Maxwell's equations \eqref{Intro_Maxwell} are linear, 
 with $\VJ_{\!\textup{total}}=0$ and $\rho=0$,
 it is sufficient to consider just a single mode
\begin{align}
E_x(t,z)=E_0\exp({-i\omega t + i k z})
\quadtext{where}
k\,,\omega\in\Cmpx
,
\label{NarrowBA_psi}
\end{align}
 with  $\Vk$ oriented along the $z$-axis (along $\Vzhat$),
 and $\VE$ along $\Vxhat$.

Now consider
 a temporally dispersive medium where the CRs
 specify permittivity $\FTepsilon(-\omega)$ and
 permeability $\mu_{\tinyzero}$,
 and
 where the ``$-\omega$'' 
  is a consequence of choosing $e^{-i\omega t}$ in \eqref{NarrowBA_psi}.
If $\omega$ and $\FTepsilon(-\omega)$ are both real
 then $\FTepsilon(-\omega)$  can be replaced with $\FTepsilon(\omega)$;
 but this is not allowed in our following calculations,
 because extra care must be taken when
 using complex permittivity to model damping. 
From Maxwell \eqref{Intro_Maxwell}
 we obtain a dispersion relation,
 and define the
 refractive index\cite{Nistad-S-2008pre,Kinsler-2009pra,Feigenbaum-KO-2009oe}.
These are
\begin{align}
  k^2 - \omega^2\ \mu_{\tinyzero} \,
  \FTepsilon(-\omega)
&=
 0
,
\label{Intro_DR}
\\
  \FTn(-\omega)^2
&=
  \cosq \,\mu_{\tinyzero} \,
  \FTepsilon(-\omega)
,
\label{Intro_def_n}
\end{align}
where
 $\co=(\epsilon_{\tinyzero}\mu_{\tinyzero})^{-1/2}$ is the vacuum
 speed of light.

We now ask whether these CRs
 correspond to a passive lossy medium,
 i.e. one dampened with no external energy added.
Given that 
 both $\omega$ and $k$ can be either real or complex, 
 there are two possibilities that are straightforward to consider.
These fit into the temporally propagated and spatially propagated
 viewpoints respectively \cite{Kinsler-2014arXiv-negfreq},
 and are:

First, 
 if $\omega$ is real and positive,
 we require that plane waves are spatially evanescent
 in the propagation direction, 
 which implies $\Im\big(\FTepsilon(-\omega)\big)>0$.

Second,
 if $k$ is real, 
 then we need $\Im(\omega)<0$ to damp the field,
 leading to the requirement that
 when
 $\omega^2\FTepsilon(-\omega)$ is real and positive,
  then $\Im(\omega)<0$.

How the fields
 represented by these modes,
 change as they cross a temporal boundary
 will depend on 
 {how the change in CRs is specified,
 and on}
 the chosen TBCs.

\section{Temporal Boundary Conditions}
\label{ch_TemporalBC}
Different types of electromagnetic TBC can be identified, 
 depending on the material response
 \cite{Xiao-MA-2014ol,Bakunov-M-2014ol,Bacot-LEFF-2016np}; 
 here we summarise focussing on
 how
 \emph{bound currents} represent the material response.
First,
 as depicted in Fig. \ref{fig_TempSurf},
 we could identify 
 the sudden change in the CRs at $t=\tzero$
 as leading to an instantaneous \emph{temporal-surface} dipole current
 ($\VJsurface$).
The total current, 
 $\VJtotal$, 
 in the medium is
\begin{align}
  \VJtotal
&=
  \VJusual - \delta(t-{\tzero})\,\VJsurface
,
\label{Intro_Jtime}
\end{align}
where $\VJusual$ is the usual finite current in Maxwell's equations:
\begin{equation}
\begin{gathered}
\nabla\cdot\VB=0
\,,\quad
\nabla\cdot\VD=\rho
\\
\nabla\times\VE + \partial_t{\VB}=0
\quadand
\nabla\times\VH - \partial_t{\VD} = \VJtotal
.
\end{gathered}
\label{Intro_Maxwell}
\end{equation}
Since the fields $\VE,\VB,\VD,\VH$ may be discontinuous we
write
\begin{align}
  \VD(t,\Vx)
&=
  \theta({\tzero}-t)\,\VD_\Xbfore(t,\Vx) + \theta(t-{\tzero})\,\VD_\Xafter(t,\Vx)
\label{D_def}
\end{align}
where $\VD_\Xbfore$ and $\VD_\Xafter$
 are the $\VD$ before and after the transition
 and $\theta$ is the Heaviside function.
Using \eqref{Intro_Jtime} and \eqref{D_def}
 in the Maxwell-Ampère equation,
 we have
\begin{align}
  \VJtotal
&=
  \VJusual
 -
  \delta(t-{\tzero})
  \VJsurface
=
  \nabla \times \VH
 -
  \partial_t{\VD}
\nonumber
\\
&=
  \theta({\tzero}-t)
  \nabla \times \VH_{\Xbfore}(t,\Vx)
 +
  \theta(t-{\tzero})
  \nabla \times \VH_{\Xafter}(t,\Vx)
\nonumber
\\
&\qquad
 -
  \theta({\tzero}-t)
  \partial_t \VD_{\Xbfore}(t,\Vx)
 -
  \theta(t-{\tzero})
  \partial_t \VD_{\Xafter}(t,\Vx)
\nonumber
\\
&\qquad
 -
  \delta(t-{\tzero})
  \big\{
    \VD_{\Xafter}(t,\Vx)
   -
    \VD_{\Xbfore}(t,\Vx)
  \big\}
  .
\nonumber
\end{align}
 This approach can also be used for $\VB$
 in the Maxwell-Faraday equation 
  to derive a similar result.
The TBC
 are
\begin{align}
  [\VD]=\VJsurface
\qquadand
  [\VB]=0
,
\label{Intro_Bdd_cond}
\end{align}
 where
 $[\VD]=\VD_{\Xafter}(\tzero,\Vx)-\VD_{\Xbfore}(\tzero,\Vx)$, etc.
One option \cite{Morgenthaler-1958irs,Kiasat-PEN-2018cleo}
 is to set $\VJsurface=0$, 
 to obtain the
 \emph{natural} TBC,
 i.e.
\begin{align}
[\VD]=0 \qquadand [\VB]=0
.
\label{TemporalBC_Natural}
\end{align}

Alternatively, 
 if treating $\VE$ and $\VB$ as the only physical fields,
 we get the \emph{fundamental} TBC,
 i.e.
\begin{align}
  [\VE]=0
\qquadand
  [\VB]=0
,
\label{Intro_Bdd_cond_E0}
\end{align}
which rely on \eqref{Intro_Bdd_cond}
 to calculate the conserved temporal-surface current $\VJsurface$.
This is an analogous approach to that describing
 the surface current around a permanent magnet
 \cite{Gratus-KM-2019ejp-dhfield}.

\section{Constant Complex CR gives Unphysical Results}
\label{ch_Constant}
%
%
 Even a time boundary between
 a vacuum with $\epsilon = \epsilon_{\tinyzero}$
 and a {lossy} medium with $\epsilon=\epsilonII$, 
 where $\epsilonII$ is a non real constant 
 with $\Im(\epsilonII)<0$  and $\Re(\epsilonII)>0$, 
 results in a failure.
For simplicity, 
 we set $\tzero=0$, 
 so that 
 $\VD(t,\Vx)=\epsilon_\tinyzero\,\VE(t,\Vx)$ for $t<0$,
 and
 $\VD(t,\Vx)=\epsilonII\,\VE(t,\Vx)$ for $t>0$; 
 and then choose a field polarization
 so that $\VE = E_x(t,z)\Vxhat$,
 and  $\VB = B_y(t,z)\Vyhat$.

Pre-boundary ($t<{0}$),
 we start with a single real mode
~
\begin{align}
  E_x(t,z)
&=
  E_{\tinyzero}
  \cos({\wbfore t - \kbfore z})
,
\label{IntroConst_E_vac}
\end{align}
 with $E_{\tinyzero}\in\Real$ and $B_y(t,z) =  E_x/\co$.
 Here $\wbfore, \kbfore$ are both real and positive
 and satisfy the vacuum dispersion relation
 $\cosq \kbfore^2 =\wbfore^2$.
Since the post-boundary lossy material, 
 with $\wafter$ and $\kafter$, 
 must have $\Im(\epsilonII) > 0$,
 the $t>{0}$ general solution is
\begin{equation}
\begin{aligned}
  E_x(t,z)
&=
  \left(
    \Amp \,e^{-i\wafter t + i \kafter z}
   +
    \Amm \,e^{-i\wafter t - i \kafter z}
   \right.
\\
&\qquad
   +
   \left.
    \Apm \,e^{i\wafter t - i \kafter z}
   +
    \App \,e^{i\wafter t + i  \kafter z}
  \right)
,  
\\
  B_y(t,z)
&=
   \frac{{\kafter}}{{\wafter}}
   \left(
     \Amp \,e^{-i\wafter t + i \kafter z}
    -
     \Amm \,e^{-i\wafter t - i \kafter z}
   \right.
\\
&\quad\qquad
   +
   \left.
     \Apm \,e^{i\wafter t - i \kafter z}
    -
     \App \,e^{i\wafter t + i \kafter z}
  \right)
,
\end{aligned}
\label{IntroConst_II_D}
\end{equation}
with the dispersion relation
\begin{align}
  {\kafter} \co 
=
  {\nafter} {\wafter}
\qquadtext{where}
  \nafter 
 =
  c_{\tinyzero}(\epsilonII\mu_{\tinyzero})^{1/2}
.
\label{IntroConst_II_DR_n}
\end{align}
The choice of root for ${\nafter}$  is unimportant,
 as both roots are included in \eqref{IntroConst_II_D}. 
The first root $\Re(\nafter)>0$, 
 since $\Im(\epsilonII)>0$, requires $\Im(\nafter)>0$;
 and the second root  $\Re(\nafter)<0$
 with $\Im(\nafter)<0$.
Applying the fundamental TBC \eqref{Intro_Bdd_cond_E0},
 just before the time boundary,
 we have
\begin{align}
  E_x(0^-,z)
&=
  \tfrac12{E_{\tinyzero}}
  \left( e^{i {\kbfore} z} + e^{-i {\kbfore} z} \right)
,
\label{IntroConst_bdd_I}
\end{align}
 and $B_y(0^-,z) = ({\kbfore}/{\wbfore}) E_x$.
Just after the time boundary
\begin{equation}
\begin{aligned}
  E_x(0^+,z)
&=
    (\Apm+\Amm) \,e^{- i {\kafter} z}
   +
    (\App+\Amp) \,e^{i {\kafter} z}
,
\\
  B_y(0^+,z)
&=
  \frac{{\kafter}}{{\wafter}} \!
  \left[
   (\Apm{-}\Amm) \,e^{- i {\kafter} z}
  +
   (\Amp{-}\App) \,e^{ i {\kafter} z}
  \right]\!
.
\end{aligned}
\label{IntroConst_bdd_II}
\end{equation}
Note that except for different constants,
 the result is the same as would be obtained
 using the natural TBC \eqref{TemporalBC_Natural}.

From \eqref{IntroConst_bdd_I} and \eqref{IntroConst_bdd_II}
 we see that ${\kafter}={\kbfore}$ or ${\kafter}=-{\kbfore}$ must hold and
$\wafter=\nbfore\,\wbfore/\nafter$.
Because we can choose ${\kafter}={\kbfore}$ without affecting the analysis,
 we can now rearrange the four unknowns to get
\begin{align}
 \Apm{=}\Amp
{=}
  \tfrac{1}{4}{E_{\tinyzero}}
  (1{+}\nafter)
\quadand
  \Amm{=}\App
{=}
  \tfrac{1}{4}
  {E_{\tinyzero}} 
  \left( 1 {-} {\nafter} \right)
.
\label{IntroConst_bdd_soln}
\end{align}
In general,
 for $t>{0}$
 all the $g^\pm_\pm$ coefficients of $E_x(t,z)$ are non zero.
Let $\cafter = \cR + i \cI=  \co /{\nafter}$ then $\cR>0$ and $\cI<0$,
 and expand \eqref{IntroConst_II_D} to yield
\begin{equation}
\begin{aligned}
  E_x(t,z)
 &=
    \Amp \,e^{i {\kafter}(-\cR t + z)} e^{ {\kafter}\cI t}
   +
    \Amm \,e^{i {\kafter}(-\cR t - z)} e^{ {\kafter}\cI t}
\\&~~~
   +
    \Apm \,e^{i {\kafter}(\cR t - z)} e^{- {\kafter}\cI t}
   +
    \App \,e^{i {\kafter}(\cR t + z)} e^{- {\kafter}\cI t}
.
\end{aligned}
\label{D_II_alt}
\end{equation}
Now we have $ {\kafter}={\kbfore}>0$ and $\cI<0$,
 so $E$ increases exponentially with time
 \emph{despite} this being a {lossy} medium.
Clearly this is physically invalid,
 so the constant complex CR model has \emph{failed}.
This has a crucial significance
 for any technology relying on temporal boundaries.
For example, 
 consider an EM wave propagating in an engineered material
 based on acrylonitrile butadiene styrene (ABS)
 whose permittivity is changed at time $t_0$
 from that of ABS ($\epsilon= 2.55 + i 0.007$)
 to $\epsilon= 2 + i0.008$;
 according to a constant complex permittivity model
 the amplitude of the EM wave is predicted to amplify exponentially,
 even though the material stays lossy.
Further, 
 substituting \eqref{IntroConst_bdd_soln} into \eqref{D_II_alt}
 we observe that in general $E_x(t,z)$ is complex valued,
 despite the TBC \eqref{Intro_Bdd_cond_E0}
 and the initial field \eqref{IntroConst_E_vac} being real-valued.
This is because the wave equation for the medium 
 when $t>0$, 
 from \eqref{Intro_Maxwell} and $\VJ_{\!\textup{total}}=0$,
 is
\begin{align}
\nabla^2\VE - \epsilonII\ddot\VE = 0
,
\label{Constant_deq_E}
\end{align}
i.e. not a real equation in real unknowns; 
 a point whose significance might be missed
 if expecting to take the real part.


\section{Using a Narrowband Approximation}
\label{ch_NarrowBA}
Using a narrowband approximation 
 we can recover the correct physical behaviour, 
 although  
 when solving the dispersion \eqref{Intro_DR}
 one needs to be careful about the square root. 
We choose $\omega, k$ in \eqref{NarrowBA_psi}
 to satisfy the dispersion relation
\begin{align}
k\, \co - \FTn(-\omega)\,\omega = 0
,
\label{NarrowBA_Dispersion}
\end{align}
and consider all relevant frequencies.
With an over-bar denoting complex conjugates, 
 the associated solutions are
 created by substituting  $\omega\to \pm{\omega}$,
 $\omega\to \pm\cnj{\omega}$,
 $k\to \pm k$,
 and
 $k\to \pm \cnj{k}$.
Then, 
 combining \eqref{Intro_DR} and \eqref{Intro_def_n}
 gives $c_\tinyzero^2k^2-\omega^2\FTn(-\omega)^2=0$.
The eight roots of this and its
 complex conjugate are given by
\begin{align}
  E_x=
  &\exp({-i\omega t \pm i k z})
  &&\text{\!\!satisfies \ }
{\FTn(-\omega)\,\omega    = \pm\co k}  ,
\label{eqn-table-first}
\\
  E_x=
  &\exp({i\omega t \mp i k z})
    &&\text{\!\!satisfies \ }
{\FTn(\omega)\,\omega = \pm\co k},
\label{eqn-table-second}
\\
  E_x=
  &\exp({-i\cnj{\omega} t \pm i \cnj{k} z})
    &&\text{\!\!satisfies \ }
{\cnj{\FTn(\omega)}\,\cnj{\omega} = \pm  \co\cnj{k} } ,
\label{eqn-table-third}
\\
  E_x=
  &\exp({i\cnj{\omega} t \mp i \cnj{k} z})
    &&\text{\!\!satisfies \ }
{\FTn({\cnj\omega})\,\cnj{\omega} =   \pm\co\cnj{k} }
,
\label{eqn-table-fourth}
\end{align}
 (see Supplementary Material) 
 which uses 
 ${\FTn({-\cnj\omega})} = \cnj{\FTn({\omega})}$, 
  a consequence of the reality condition 
  on $\FTn$.

Now, 
 if we apply the NBA, 
 we know that 
 all the fields are concentrated
 about two modes, 
 i.e. at $\omega \approx \omega_{\tinyzero}$
 and $\omega \approx \cnj\omega_{\tinyzero}$.
Let $\nconst=\FTn(-\omega_{\tinyzero})$, 
 so that for 
 $\omega\approx\omega_{\tinyzero}$ we have
 $\FTn(-\omega)\approx\nconst$.
Similarly, 
 we also have
 $\cnj\nconst = \FTn(\cnj\omega_{\tinyzero}) $.  

To obtain a real $E_x(t,z)$ we need to add the complex conjugate.
As $\Im(\omega_{\tinyzero})=\Im(-\cnj\omega_{\tinyzero})$
 we construct the total field from \eqref{eqn-table-first}
 and \eqref{eqn-table-fourth} above to give
\begin{align}
  E_x(t,z)
&=
  \Amp\, e^{-i\omega t + i k z}
 +
  \Amm\,e^{-i\omega t - i k z}
\nonumber
\\
&\quad
 +
  \Apm\,e^{i\cnj{\omega} t - i \cnj{k} z}
 +
  \App\,e^{i\cnj{\omega} t + i \cnj{k} z}
.
\label{Calc_Gen_D}
\end{align}

We are now in a position
 to reconsider our time boundary system
 in the context of dispersive media. 
Before the time boundary $t<\tzero=0$ we have the vacuum. 
Assuming a narrowband initial pulse
 enveloping the single mode \eqref{IntroConst_E_vac},
 where $\wbfore=\co\kbfore\in\Real$, 
 and $\kbfore>0$. 
Whichever TBC we assume we obtain $\kafter=\pm\kbfore$. 
Again the choice of root is unimportant,
 therefore we set $\kafter=\kbfore$,
 so that $\kafter>0$. 
After the time boundary,
 we have dispersion relations given by
 \eqref{eqn-table-first}--\eqref{eqn-table-fourth}. 
In order to obtain the physical solutions
 consistent with $\kafter>0$
 we choose $\Re(N)>0$,
 $\Im(N)>0$ and the modes
 given by \eqref{eqn-table-first} and \eqref{eqn-table-fourth}.
For convenience we define a wave speed $\CC=\CR+i\CI=\co/\nconst$,
 where $\CR,\CI\in\Real$, $\CR>0$ and $\CI<0$.
From \eqref{Calc_Gen_D} we have,
\begin{equation}
\begin{aligned}
&E_x(t,z){=}
 E_{\tinyzero}\, e^{k\CI t}
\Big(
\frac{\nconst{+}1}{\nconst{+}\cnj{\nconst}} \,e^{-i k(\CR t - z)}
+
\frac{\cnj\nconst{+}1}{\nconst{+}\cnj{\nconst}}\,e^{-i k(\CR t + z)}
\Big)
\end{aligned}
\label{Calc_Gen_E_expand}
\end{equation}
where $k=\kafter>0$
(see Supplementary Material), 
 and the EM field remains real and damped.
However,
 we no longer have a single
 differential equation \eqref{Constant_deq_E},
 as we now need two ($\nconst$, $\cnj\nconst$):
\begin{equation}
\begin{aligned}
\begin{cases}
\nabla^2\VE - {\nconst}^2 \ddot{\VE} = 0
&\text{ for frequencies $\omega\approx\omega_{\tinyzero}$ in $\VE$ },
\\
\nabla^2\VE - {\cnj\nconst}^2\ddot{\VE} = 0
&\text{ for frequencies $\omega\approx\cnj{\omega_{\tinyzero}}$
in $\VE$ }
.
\end{cases}
\end{aligned}
\label{Calc_NBA_PDE}
\end{equation}
So by using the NBA with dispersion,
 we can match the TBC,
 and obtain a physical, 
 dampened, 
 real-valued solution for the electric field; 
 but the two post-boundary refractive indices required
 suggest that more appropriate models are necessary.

\def\Eng{{\cal E}}

\newcommand\Tsub[2]{{#1}_{#2}}
\def\Etot{\Tsub{\Eng}{\text{total}}}
\def\Eem{\Tsub{\Eng}{\text{EM}}}
\def\Eee{\Tsub{\Eng}{\!E}}
\def\Ebb{\Tsub{\Eng}{\!B}}
\def\Ejb{\Tsub{\Eng}{\!J_{\text{b}}}}

\begin{figure}
\resizebox{0.495\columnwidth}{!}{\includegraphics{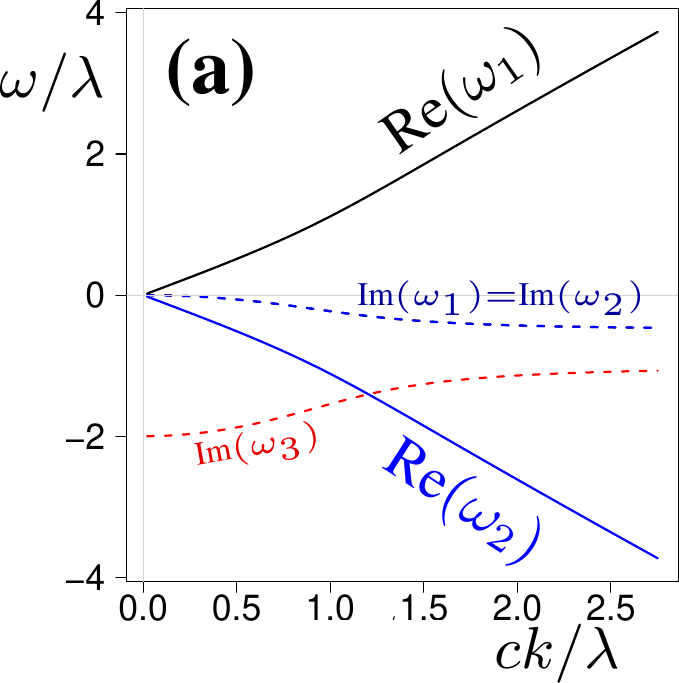}}
\resizebox{0.455\columnwidth}{!}{\includegraphics{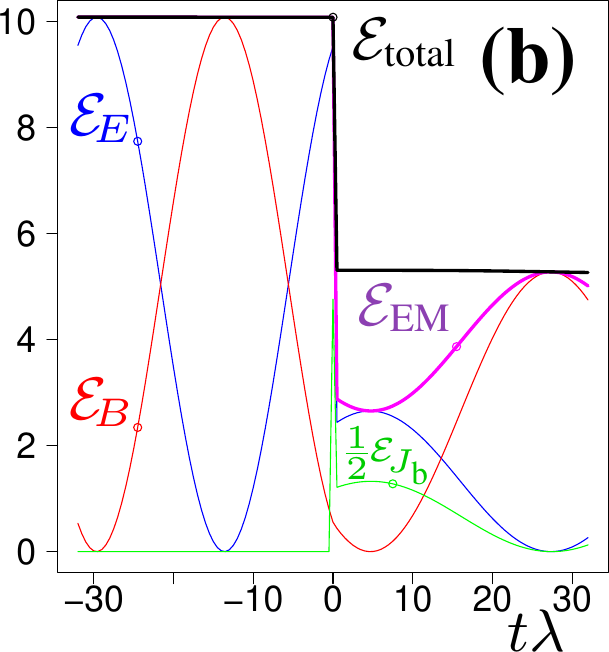}}\\
~~~
\resizebox{0.455\columnwidth}{!}{\includegraphics{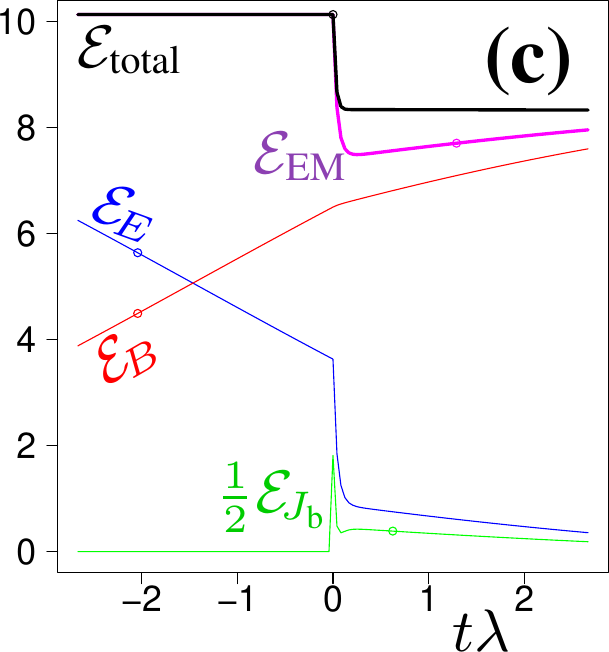}}
\resizebox{0.455\columnwidth}{!}{\includegraphics{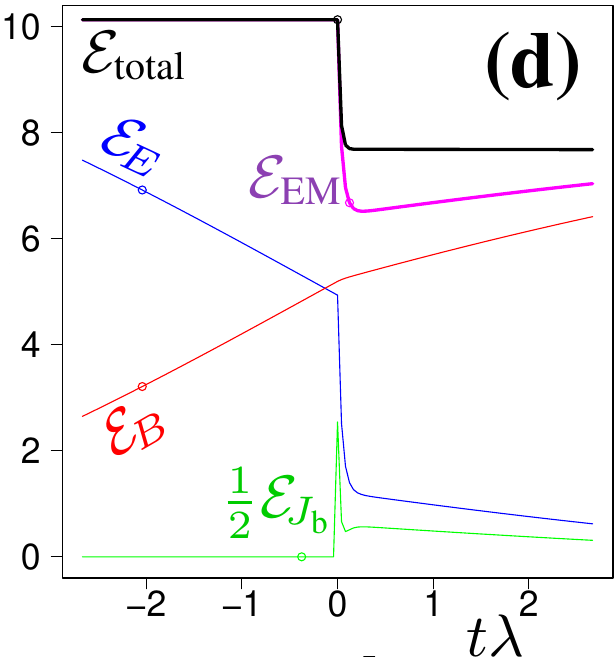}}
\caption{
(a) The dispersion relation \eqref{ABC_DR}
 for the simple material response model of \eqref{ABC_DE_P}, 
 where
 $\chi_{\protect\tinyzero}=\epsilon_{\protect\tinyzero}\lambda$
 and real $k$.
The two propagating solutions $\omega_1$, $\omega_2 \in\Cmpx$
 have the same small damping,
 while the third $\omega_3 \in \imath \Real$ is pure loss.
In (b,c,d) we see what happens
 as the system crosses a time boundary at $t=0$
 between
 vacuum $\chi_{\protect\tinyzero}=0$
 and 
 $\chi_{\protect\tinyzero}=\epsilon_{\protect\tinyzero}\lambda$; 
 where (c,d) show a narrower range of $t$ around $t=0$.
The different initial conditions used in (b,c,d)
 demonstrate the timing-sensitive behaviour caused by the boundary.
These graphs compare the energies in the three fields
 $\Eee=\epsilon_{\protect\tinyzero} E^2$,
 $\Ebb=B^2/\mu_{\protect\tinyzero}$
  and
 $\Ejb=J_{\text{b}}^2/(\chi_0\lambda)$; 
 but we show $\tfrac{1}{2}\Ejb$ to increase visibility.
Total energy is conserved for $t<0$; 
 but just after the time boundary,
 $\Eee$ (and hence $\Eem$)
 rapidly reduce as excitation is transferred to $J_{\text{b}}$,
 where it is strongly damped.
The energy lost though $J_{\text{b}}$ depends on 
 how strongly it is driven by $E$
 until the new dynamic near-equilibrium is reached.
Direct damping from the lossy $\omega_3$ solution
 is shown by \emph{the rapid and significant decay} of $\Ejb$
 just after the transition.
}
\label{fig-minimal-time}
\end{figure}

\def\epszero{\epsilon_{\tinyzero}}
\def\chizero{\chi_{\tinyzero}}

\section{Dynamic Material Models}
\label{ch_ABC}

It is not possible to both implement a 
 physically consistent time boundary
 with a constant complex $\epsilonII$, 
 and escape the requirement for the NBA in \eqref{Calc_NBA_PDE}.
This means
 we must instead use an explicitly causal \cite{Kinsler-2011ejp}
 dynamic response model
 for the medium,
 thus providing fully dispersive CR.
Since a response model adds extra field(s)
 to describe the material response, 
 the coupled EM-material system
 will both have more than two modes and need ABCs.
Essentially, 
 defining a material response model immediately creates a demand
 for ABCs -- 
 they are simply the boundary conditions on the auxilliary fields 
 (such as polarization $\VP$ or bound current $\VJ_{\textup{b}}$)
  that one has decided to use.

A minimal but sufficient material response can be given by
\begin{align}
  \VD   =   \epszero  \VE + \VP
\quadtext{where}
  \dot{\VP} 
=
 -
  \lambda \VP
 +
  \chizero \VE
,
\label{ABC_DE_P}
\end{align}
so that in the steady state $\VP = \chizero  \lambda^{-1} \VE$.
In this model, 
 as long as the loss $\lambda$ is large compared to the field frequency, 
 with the desired change
 in permittivity $\Delta \epsilon = \chizero /  \lambda$ held fixed, 
 it indeed responds as if it were a medium of complex constant CR.
This model also requires a boundary condition
 for the dielectric polarization $\VP$ field,
 namely
 $[\VP]=0$.

However, 
 it is significant that
 the polarization field \eqref{ABC_DE_P}
 is in fact derived from a bound current
 $\VJ_{\textup{b}} = \dot\VP = - \lambda\VP
                   + \chizero\VE$.
This $\VJ_{\textup{b}}$
 is driven by, 
 and hence indirectly applies loss to
 the field $\VE$; 
 and our simulation results
 shown in figure \ref{fig-minimal-time}
 use this $\VJ_{\textup{b}}$ approach.
Used after the time boundary ($t>{\tzero}$),
 this medium
 has a dispersive behaviour given by
\begin{align}
  \FTepsilon(-\omega)
&= 
  \epsilon_{\tinyzero}
 +
  \frac{\chi_{\tinyzero}}
       {\lambda-i\omega}
,
\label{ABC_epsilon}
\end{align}
where $\chi_{\tinyzero}>0$ and $\lambda>0$; 
 and the steady state is reached when $\omega/\lambda \rightarrow 0$. 
The resulting cubic dispersion relation is
\begin{align}
 \epsilon_{\tinyzero} \mu_{\tinyzero}
  \left( \lambda - i \omega \right) 
  \omega^2
 +
  \mu_{\tinyzero} \chi_{\tinyzero} \omega^2
 -
  k^2 \left( \lambda - i \omega \right)
=
  0
.
\label{ABC_DR}
\end{align}
This equation in $\omega$ and fixed  $k^2$
 produces three solutions
 as opposed to the two
 given in \eqref{IntroConst_II_D} or \eqref{Calc_Gen_E_expand}.
Since \eqref{ABC_DR} has real coefficients
 when written as a polynomial in $(i\omega)$,
 and 
 as a cubic has at most two non-real roots for $i\omega$, 
 the three solutions consist of a complex conjugate pair
 and a real valued one.
The pair correspond to counter propagating waves (modes)
 with the same damping,
 with the other solution (mode) being non-propagating and purely damped
 (see figure \ref{fig-minimal-time}(a)).

At a time boundary where $k\in\Real$ and $k>0$,
 we have three outgoing modes, 
 requiring three boundary conditions.
Two TBC are given by either the natural \eqref{TemporalBC_Natural}
 or fundamental \eqref{Intro_Bdd_cond_E0} TBC \eqref{Intro_Bdd_cond}; 
 the ABC of course is just the $[\VP]=0$
 evident from the dynamic model \eqref{ABC_DE_P}.
This ABC is analogous to the Pekar ABC \cite{Pekar-1958zetf},
 which are required when an EM wave
 passes into a spatially dispersive 
 medium \cite{AgranoGinsberg,Kinsler-2021pnfa-spatype}.

Simulation results are shown
 in figure \ref{fig-minimal-time}(b,c,d), 
 demonstrating the system behaviour as it passes the boundary.
Despite the excellent match to a medium with constant complex CR
 before and sufficiently far after the boundary transition, 
 it does not exhibit the unphysical behaviour of
 prescribed constant complex CRs.
Instead,
 just after $t=t_b=0$ in
 figures \ref{fig-minimal-time}(c,d)
 we can see
 a rapid rebalancing as $\VE$ and $\VJ_{\textup{b}}$
 (i.e. $\dot{\VP}$)
 adjust to the recently changed $\chizero$.

Note also the overshoot in  $\Eng_J$ 
 just after the boundary, 
 which can be attributed
 to the dampened $\omega_3$.
These occur on a timescale set by $\lambda$,
 and incur an $\VE$ dependent energy loss.
From a dynamic perspective, 
 system appears to have two loss processes, 
 although both are in fact different manifestations of the same 
 $\lambda$ loss term.
In a steady-state medium,
 i.e. away from the boundary, 
 the losses are gradual and proportional to $\omega/\lambda$.
This counter-intuitive dependence on the $1/\lambda$
 is because the loss depends on the mismatch
 between the ideal $\VJ_{\textup{b}}$ (or polarization $\VP$)
 and the actual value; 
 but for larger $\lambda$ values,
 the mismatch becomes smaller.
In contrast, 
 as the system transitions across the time boundary, 
 the sudden change in $\chizero$ means that the $\VE$-dependent mismatch 
 can suddenly become very large, 
 and this causes equally large and rapid losses
 as $\VJ_{\textup{b}}$ (or $\VP$)
 re-synchronise to the electric field $\VE$.
However, 
 note that in the special case where $\VE$ is zero at the boundary, 
 the mismatch remains small and no significant rapid loss takes place.

\def\Kfield{{\boldsymbol K}}

A key further point of interest
 is whether there are any side-effects if the model parameters
 $\lambda$ and $\chizero$ change with time.
On the basis of \eqref{ABC_DE_P} this does not appear
 to be the case, 
 since the only time derivative term is that on the LHS, 
 applied to $\VP$.
However, 
 to check this properly we need to reformulate the model
 so that it is explicitly based on bound currents,
 which are the true microscopic physical property.
Taking the time derivative of \eqref{ABC_DE_P}
 and rearranging (see Appendix \ref{S-minimaldetail})
 leads to the model equation
~
\begin{align}
  \dot{\Kfield}
&=
   -
  \lambda
  \left[
    \Kfield
   +
    \chizero \VE  
  \right]
 ~
 +
 ~
  \frac{\dot{\lambda}}{\lambda}
  \Kfield
,
\label{eqn-trivial-dK0}
\end{align}
 where $\Kfield = \VJ_{\textup{b}} - \chizero \VE$
 is an offset for the bound current $\VJ_{\textup{b}}$.
This construction has ensured that there is only one
 time derivative applied to one field quantity
 (i.e. $\Kfield$), 
 so that it retains an unambiguously causal form \cite{Kinsler-2011ejp}; 
 it is in fact this equation we integrate, 
 along with Maxwell's equations,
 to get the results
 shown in Fig. \ref{fig_TempSurf}(b-d)
We can see from the equation that the value of $\Kfield$
 is going to be continually trying to catch up to 
 the present value (albeit scaled) of $\VE$, 
 on a timescale set by $\lambda$.
Since the introduction of $\Kfield$ removes any dependence on 
 the time derivative of $\chizero$,
 a temporal boundary in the $\chizero$ value is implemented simply 
 by changing the parameter $\chizero$ as we integrate step by step.
However, 
 extra care need to be taken if $\lambda$ changes
 at a temporal boundary; 
 since its time derivative $\dot{\lambda}$ also appears.
This would either be specified as part of the simulation environment
 defining $\chizero$, $\lambda$, $\dot{\lambda}$; 
 or we might add an auxilliary equation for $\lambda$
 if it had its own temporal dynamics.



This model (i.e. \eqref{ABC_DE_P}, or \eqref{eqn-trivial-dK0}), 
 being defined by a temporal differential equation, 
 is necessarily explicitly causal \cite{Kinsler-2011ejp}.
It works as intended, 
 i.e. to create a near-constant permittivity, 
 when 
 (i) 
 the desired positive permittivity shift $\chi_{\tinyzero} > 0$, 
 and when
 (ii)
 $\chizero/\lambda \epszero \lessapprox 3$.
In this regime the polarization (or microscopic polarization current)
 adiabatically follows the phase of electric field, 
 and so accurately models the desired effective permittivity; 
 indeed the effective loss at low frequencies
 is proportional to $\omega/\lambda$, 
 i.e. increasing the polarization current loss $\lambda$ 
 actually \emph{reduces} the effective damping
 in the CW limit.
However, 
 if the first condition does not hold,
 a positive polarization can have a negative energy; 
 thus the amplitudes of $\VE$ and $\VP$ 
 can increase without limit
 whilst still conserving energy. 
Alternatively, 
 if the parameters change so that the second condition starts to fail, 
 the three modes --
 two electromagnetic and one polarization --
 become ever more strongly coupled and 
 eventually exhibit a complicated dynamics 
 (and dispersion) 
 not relevant to our presentation here.

\section{Beyond the minimal dynamic model}
\label{S-dynoplus}

The minimal model used above 
 performs well, 
 and has the considerable advantage of having
 a simple behaviour,
 thus clarifying the general principles for
 handling time boundaries.
However, 
 it is not typical of the material models used in practical situations
 which are often based on Drude or Lorentz oscillators
 that consider the polarization $\VP$.
Consequently, 
 we now consider a more general situation
 by using the result that any causal response 
 can be expressed as a sum of Lorentz responses \cite{Dirdal-S-2013pra}.
Working in the time domain, 
 we have  
 $\VD=\epsilon_{\tinyzero} \VE+\sum_{\sind}\VP_\sind$.
With $\sind_{\max}$ being the number of oscillators,
 we have then $\sind_{\max}$ second order equations:
\begin{align}
\ddot{\VP}_\sind + \left(\lambda_\sind{\VP}_\sind\right)\dot{\,} + \alpha_\sind \VP_\sind
  = \chi_\sind \VE
.
\label{ABC_2nd_order}
\end{align}
Here $\chi_\sind$ is the coupling,
 $\lambda_\sind$ is the damping,
 and $\alpha_\sind$ the natural frequency of the oscillator.
Note that the time derivative in the second LH term 
 acts on the product $\lambda_\sind {\VP}_\sind$, 
 not just on ${\VP}_\sind$.
Since the Lorentian oscillators follow second order dynamical equations, 
 they will therefore each require two extra boundary conditions, 
 for a total of $2\sind_{\max}$ ABCs.
The natural ABCs for the $\sind$-th oscillator are
\begin{align}
[\VP_\sind]=0
\qquadand
[2\dot{\VP}_\sind+\lambda_\sind\,\VP_\sind]=0
.
\label{ABC_2nd_order_ABC}
\end{align}
To see this we substitute $
\VP(t,\Vx)
=
  \theta({\tzero}-t)\,\VP_\sind^\Xbfore(t,\Vx) +
  \theta(t-{\tzero})\,\VP_\sind^\Xafter(t,\Vx)
$ and likewise for $\lambda_\sind$, $\alpha_\sind$ and $\chi_\sind$,
into (\ref{ABC_2nd_order}) and assume (\ref{ABC_2nd_order}) holds on
  both sides of the boundary. 
This gives 
\begin{align*}
\delta'({\tzero}-t)\,\big(\VP_\sind^\Xafter- \VP_\sind^\Xbfore\big)
+
2\,\delta({\tzero}-t)\,
\big(\dot{\VP}_\sind^\Xafter - \dot{\VP}_\sind^\Xbfore\big)
 \qquad&
\nonumber
\\
+
\delta({\tzero}-t)\,
\big(\lambda_\sind^\Xafter{\VP}_\sind^\Xafter 
- \lambda_\sind^\Xbfore{\VP}_\sind^\Xbfore\big)
&=0.
\end{align*}
The coefficients of $\delta'({\tzero}-t)$ and $\delta({\tzero}-t)$
must both be zero. This gives (\ref{ABC_2nd_order_ABC}).
Observe that if $[\lambda_\sind]=0$ then (\ref{ABC_2nd_order_ABC})
reduces to $[\VP_\sind]=0$ and
$[\dot{\VP}_\sind]=0$.
This result explains why in \eqref{ABC_2nd_order}
 that $\lambda_\sind$ is inside the time derivative in the second LH term.
If it had been outside the derivative, 
 then if $[\VP_\sind]\ne0$, 
 the alternative $\lambda_\sind({\VP}_\sind)\dot{}$ 
 would contain the product $\delta(t-{\tzero})\theta(t-{\tzero})$, 
 which is not defined mathematically.
An alternative method for setting out these ABCs
 would be to factorise \eqref{ABC_2nd_order}
 into two first order pieces; 
 and it is also possible to follow either type of analysis
 for a bound current $\VJ_{\textup{b}}$
 \cite{Kinsler-2017arXiv-dehydro}.

We can also calculate the number of ABCs required
 in the frequency domain,
 using the case of no boundary, 
 or when the system is very far from the boundary.
Here the parameters $\chi_\sind$, 
 $\lambda_\sind$
 can be treated as constants, 
 where
 so that
 \eqref{ABC_2nd_order} gives us the dispersion 
\begin{align}
\FTepsilon(-\omega)
=
\epsilon_{\tinyzero} +
\sum_{\sind=1}^{\sind_{\max}} \frac{\chi_\sind}{-\omega^2 -
\imath\lambda_\sind \omega + \alpha_\sind^2}
.
\label{ABC_epsilon_sum}
\end{align}
Expanding out the dispersion relation resulting from \eqref{Intro_DR}
 then
 leads to a polynomial in $\omega$ of degree $(2\sind_{\max}+2)$, 
 and hence $(2\sind_{\max}+2)$ modes.
Matching these modes across a boundary between 
 different materials will then require $(2\sind_{\max}+2)$ 
 boundary conditions, 
 $2\sind_{\max}$ of which are ABCs.

\section{Conclusion}

We have shown that even simple time boundaries in optics
 cannot be described by
 the standard ``constant complex permittivity'' model.
Only with a dynamic or
 dispersive model of the propagation medium
 can physical results be predicted.
This conclusion is supported by 
 NBA calculations and time domain simulations, 
 and is easily generalized
 to other wave systems.
It has significant implications for the design and modelling of 
 both experiments and applications of future technologies
 based on temporal boundary phenomena.
Our conclusion may be particularly relevant to the 
 propagation of relativistically intense electromagnetic waves
  in plasma, 
 the propagation of relativistic plasma waves
  in laser wakefield accelerators in the bubble regime, 
 and
 for relativistic ionisation fronts and relativistically induced
 transparency in laser-solid interactions.

It is arguably unsurprising that a good model of a time boundary
 requires a model that can admit non-trivial time dependence, 
 i.e. either a time domain response model
 or a frequency domain dispersion.
Time boundaries are a temporal,
 dynamic phenomena,
 and need to be treated as such.


%
\acknowledgments

\noindent
\textbf{Acknowledgments:}
{JG and PK acknowledge support provided
by STFC (Cockcroft Institute, ST/G008248/1 and ST/P002056/1), 
and with DAJ acknowledge 
 EPSRC (Lab in a Bubble, EP/N028694/1). 
DAJ also acknowledges
 the EU's Horizon 2020 programme (no. 871124 Laserlab-Europe),
RS thanks AFOSR (FA8655-20-1-7002), 
 and
 PK acknowledges
 recent support from EPSRC 
 (QUANTIC, EP/T00097X/1).}
 (QUANTIC, EP/T00097X/1). 

%


\appendix


\appendix





\section{In Linear Media -- spatial evanescence}

\begin{lemma}
Given that for $\omega\in\Real$, $\omega>0$ plane waves are spatially
evanescent in the propagation direction then
$\Im\big(\FTepsilon(-\omega)\big)>0$.
\end{lemma}

\begin{proof}
Given \eqref{NarrowBA_psi} with $\omega>0$. For $\Re(k)>0$ then the
direction of propagation is positive $z$. If the plane waves are
evanescent for positive $z$ then $\Re(i k z)<0$ hence $\Im(k)>0$. This
implies $k$ lies in the top right quadrant of $\Cmpx$. Hence
$\Im(k^2)>0$.

Likewise for $\Re(k)<0$ then the
direction of propagation is negative $z$. If the plane waves are
evanescent for negative $z$ then $\Re(i k z)<0$ hence $\Im(k)<0$. This
implies $k$ lies in the bottom left quadrant of $\Cmpx$. Hence
$\Im(k^2)>0$.

 In both cases $\Im(k^2)>0$ and since $\omega^2>0$ then
\eqref{Intro_DR} implies
$\Im\big(\FTepsilon(-\omega)\big)>0$.
\end{proof}

\section{In Linear Media  --  a note on negative refractive index}


From \eqref{Intro_def_n},
  $\Im\left( \FTepsilon(-\omega) \right) > 0$ implies that
 $\FTn(-\omega)$ is either 
 (a) in the top right quadrant of the complex plane 
    $\Set{\Re(\FTn) > 0  \text{~and~}  \Im(\FTn) > 0}$, 
 or 
 (b) in the bottom left quadrant 
    $\Set{\Re(\FTn) < 0 \text{~and~} \Im(\FTn) < 0}$.
Consequently, 
 having $\Re\left( \FTn(-\omega) \right) < 0$
 does not contradict our 
 assumption
 that the real part of both permittivity and permeability are positive:
 it is still possible to have a negative index of refraction
 \cite{Boardman-KV-2005e,SchusterOPTICS}.

\section{Using a Narrowband Approximation  --  the eight roots}

\begin{proof}[Demonstration of
    \eqref{eqn-table-first}-\eqref{eqn-table-fourth}]
Since we chose \eqref{NarrowBA_psi} to satisfy
 \eqref{NarrowBA_Dispersion} then
 this is \eqref{eqn-table-first} for
 $E_x=\exp({-i\omega t + i k z})$. 
Replacing $k\to=k$ then gives
 \eqref{eqn-table-first}.

Replacing $\omega\to -\omega$ and $k \to-k$ then gives
 \eqref{eqn-table-third}.

Taking the complex conjugate of \eqref{eqn-table-first}
 and using ${\FTn({-\cnj\omega})}=\cnj{\FTn({\omega})}$
 gives \eqref{eqn-table-fourth}.

Replacing $\omega\to -\omega$ and $k \to  -k$
 in \eqref{eqn-table-fourth}
 then gives \eqref{eqn-table-third}.

\end{proof}

\section{Using a Narrowband Approximation  --  proof}

\def\nafter{N}
\begin{proof}[Proof of \eqref{Calc_Gen_E_expand}]

From \eqref{Calc_Gen_D} we have
\begin{equation}
\begin{aligned}
  B_y(t,z)
&=
  \frac{e^{k\CI t}}{\CI-i\CR}
  \Big(
    \Amp \,e^{-i k(\CR t - z)}
   -
    \Amm \,e^{-i k (\CR t + z)}\Big)
\\
& \quad + \frac{e^{k\CI t}}{\CI+i\CR}
   \Big(
    \Apm \,e^{i k (\CR t - z)}
   -
    \App \,e^{i k (\CR t + z)}
  \Big)
\\
&=
  \frac{\cnj\nafter\, e^{k\CI t}}{\co}
  \Big(
    \Amp \,e^{-i k(\CR t - z)}
   -
    \Amm \,e^{-i k (\CR t + z)}\Big)
\\
& \quad + \frac{\nafter\,e^{k\CI t}}{\co}
   \Big(
    \Apm \,e^{i k (\CR t - z)}
   -
    \App \,e^{i k (\CR t + z)}
  \Big)
\end{aligned}
\end{equation}
Using the fundamental TBC \eqref{Intro_Bdd_cond_E0}
 we get for $t=0^+$ and $k\in\Real$
\begin{align*}
  E_x(0^+,z)
&=
  e^{i k z} (\Amp + \App) + e^{-i k z} (\Amm+\Apm)
\\
  \co B_y(0^+,z)
&=
 e^{i k z}(\cnj{\nafter}\Amp -\nafter\App)
+
 e^{-i k z}(\nafter\Apm-\cnj{\nafter}\Amm )
\end{align*}
and hence
\begin{align*}
\Amp = \cnj{\Apm} =
\frac{E_{\tinyzero}}{2}
\frac{\nafter+1}{\nafter+\cnj{\nafter}}
\quadand
\Amm = \cnj{\App}=
\frac{E_{\tinyzero}}{2}
\frac{\cnj\nafter+1}{\nafter+\cnj{\nafter}}
\end{align*}
Substituting into \eqref{Calc_Gen_D}
 gives \eqref{Calc_Gen_E_expand}.
\end{proof}

\section{The minimal model and bound currents} \label{S-minimaldetail}

\def\pCurrent{\VJ_{\textup{b}}}
\def\pEfield{\VE}
\def\pPfield{\VP}
\def\pDfield{\VD}
\def\pPermittivityVac{\pPermittivity_\tinyzero}
\def\pSlave{\lambda}

The concept behind our minimal model 
 is that it should mimic the ``ideal''
 of a constant-like permittivity
 with real and imaginary components.
Thus we want there to exist a dielectric polarization $\VP$ that 
 closely follows the current value of the electric field $\VE$, 
 but allows freedom for dynamical variation
 about its chosen target value.
Most simply, 
 we can write
~
\begin{align}
  \pPderivT \pPfield
&=
 -
  \pSlave
  \pPfield
 +
  \chizero 
  \pEfield
,
\label{eqn-trivial-dP}
\end{align}
 where in the steady-state limit
 we have the desired $\pPfield = (\chizero/\pSlave) \pEfield 
 = \Delta \varepsilon \, \pEfield$.
In this model,  
 the dynamics are simply that 
 $\pPfield$ exponentially decays towards to $\Delta \varepsilon \, \pEfield$
 with rate $\pSlave$.

However, 
 from a microscopic point of view, 
 a charge or current is the more useful physical property.
Thus to adapt the initial concept in \eqref{eqn-trivial-dP}
 we use the fact that 
 $\pCurrent = \pPderivT \pPfield$.
However, 
 since substitution of $\pCurrent $ into \eqref{eqn-trivial-dP}
 leaves us with no dynamics, 
 we also apply an extra time derivative to the above equation.
As a result we have 
~
\begin{align}
  \pPderivT \pCurrent
&=
 -
  \pSlave
  \pCurrent
 +
  \pPderivT 
  \left(
    \chizero 
    \pEfield
  \right)
 -
  \left(\pPderivT  \pSlave \right)
  \pPfield
.
\label{eqn-trivial-d2P}
\end{align}
Since this has time derivatives of
 fields
 on both sides, 
 it is not straightforward to interpret it in a causal manner
 \cite{Kinsler-2011ejp}.
Thus
 we combine the $\pCurrent$ and $\pEfield$ fields, 
 both of which are subject to an applied time derivative,
 into a single quantity:
 ~
\begin{align}
  \Kfield
&=
  \pCurrent
 - 
    \chizero 
    \pEfield
.\label{eqn-trivial-Kdef}
\end{align}

Thus we rearrange \eqref{eqn-trivial-d2P}, 
 and substitute for $\pPfield$
 using a rearranged \eqref{eqn-trivial-dP}, 
 to get
~
\begin{align}
  \pPderivT
  \left(
   \pCurrent
   -
    \chizero 
    \pEfield
  \right)
&=
 -
  \pSlave
  \pCurrent
 ~
 -
 ~
  \left(\pPderivT  \pSlave \right)
  \frac{1}{\pSlave}
  \left[
    \chizero 
    \pEfield
   -
    \pPderivT
    \pPfield
  \right]
,
\\
  \pPderivT
  \Kfield
&=
 -
  \pSlave
  \pCurrent
 ~
 +
 ~
  \left(\pPderivT  \pSlave \right)
  \frac{1}{\pSlave}
  \left[
    \pCurrent
   -
    \chizero 
    \pEfield
  \right]
,
\\
  \pPderivT
  \Kfield
&=
 -
  \pSlave
  \left[
    \Kfield
   +
    \chizero 
    \pEfield
  \right]
 ~
 +
 ~
  \frac{\pPderivT \pSlave}{\pSlave}
  \Kfield
.
\label{eqn-trivial-dK}
\end{align}
 where the changes (effects) on the LHS
 from the first RHS term
  are determined solely by the known present values
 of $\Kfield$ and $\pEfield$.
If the system parameter $\pSlave$ is time dependent (i.e. $\pSlave(t)$), 
 then see also that we need
 to know $\pPderivT \pSlave$ as well as $\pSlave$; 
 although such a specification is not required for $\chizero$.
This is particularly relevant in the case
 of an (otherwise constant) medium
 with an abrupt change at a time boundary.

From \eqref{eqn-trivial-dK} we can see that the value of $\Kfield$
 is going to be continually trying to catch up to 
 the present value (albeit scaled) of $\pEfield$, 
 on a timescale set by $\pSlave$.
Clearly, 
 for this model to function as intended
 we will want the $\pSlave$ timescale to be faster than
 the largest significant frequency component of $\pEfield$.

In the harmonic case,
 we can Fourier transform the evolution equation; 
 here a time-domain field $A(t)$ 
  becomes a frequency-domain $\tilde{A}(\omega)$.
Then
 \eqref{eqn-trivial-dK} tells us that 
~
\begin{align}
  -\imath \omega \tilde{\Kfield}
&=
  - \pSlave 
  \left(
    \tilde{\Kfield} + \chizero \tilde{\pEfield}
  \right)
,
\\
\textrm{i.e.} \qquad
  \tilde{\Kfield}
&=
    \frac{- \chizero \pSlave}
         {-\imath\omega + \pSlave}
  \tilde{\pEfield}
\qquad
\textrm{or}
\qquad
  \tilde{\pCurrent}
=
  \chizero
  \left[
    \frac{- \pSlave}
         {\pSlave - \imath\omega}
   +
    1
  \right]
  \tilde{\pEfield}
.
\label{eqn-trivial-harmonic}
\end{align}
In the intended limit where $\pSlave \gg \omega$, 
 we can expand to second order
 so that
~
\begin{align}
  \tilde{\pCurrent}
&\simeq
 -
  \imath\omega
  \chizero \pSlave^{-1}
  \left[
    1
   -
    \imath\omega
    \pSlave^{-1}
   \right]
  \tilde{\pEfield}  
,
\label{eqn-trivial-harmonic-expansion}
\end{align}
and as a result
 we have
~
\begin{align}
  \tilde{\pPfield} 
&=
  \left[
    \Delta \varepsilon_r
   -
    \imath
    \Delta \varepsilon_i (\omega)
  \right]
  \tilde{\pEfield}
,
\qquad
  \textrm{where}\qquad
  \Delta \varepsilon_i (\omega)
=
  \Delta \varepsilon_r \omega / \pSlave
.
\end{align}
This means that this minimal dynamical model
 for the standard ``constant permittivity'' assumption
 which will match the target real-valued permittivity
 in the large $\pSlave$ limit, 
 with the concommitant introduction of a loss
 that gets ever smaller in the ideal large $\pSlave$ limit.

\section{Remarks  --  Causality, dispersion, and the NBA}

\vspace{-3mm} 
Across our time boundary, 
 a change in constant complex CR for a medium 
 is just an instantaneous change, 
 which is straightforwardly causal.
Causal behaviour, 
 however, 
 of itself is not necessarily guaranteed
 to give physically reasonable predictions.
Indeed, 
 all but one of the results in this paper are causal --
 even the unphysical \eqref{Constant_deq_E}, 
 where the future behaviour is by construction
 explicitly dependent on the past behaviour.
The single exception is
 the narrowband result \eqref{Calc_NBA_PDE}, 
 because under this approximation
 questions of causality are moot.

Causality is often tested by applying
 the Kramers-Kronig relations (see e.g. \cite{Kinsler-2011ejp}), 
 but they do not apply to all situations.
For example, 
 even though the result \eqref{D_II_alt} is causal, 
 and clearly so when solving in the time domain,
 as an exponentially increasing function it cannot be Fourier transformed
 so as to allow Kramers-Kronig to be tested.
Indeed, 
 since the original function of Kramers-Kronig
 was to analyse, test, or correct raw data \emph{collected}
 in the frequency domain, 
 using them as a causality test
 when a time domain description is already available is redundant.

This is why the dynamical model \eqref{ABC_DE_P} is a natural starting point
 for an examination of temporal boundaries; 
 although of course more complicated models, 
 such as the summed Lorentzians of \eqref{ABC_epsilon_sum}, 
 or ones involving (causal) integral kernels,
 can be constructed.
Notably, 
 even the highly simplified  \eqref{ABC_DE_P}, 
 designed to give results
 that in the appropriate limit 
 are
 as close as possible to the failed constant complex CR model,
 is sufficient to restore physical behaviour.

\section{Material responses: Cauchy and convolutions}

\vspace{-2mm}
In principle, 
 we could try to expressed this  time boundary situation as a Cauchy problem -- 
 i.e. namely asking what is the subsequent wave behaviour
 if the Cauchy (initial condition) data is $E(z)=E_0\cos(k_a z)$. 
However, 
 it is not clear how this extended field could be set
 up as a realistic initial condition; 
 and, 
 further, 
 trying to use a {straightforward} convolution approach
 would fail because $E(z,t)$ for $t<0$ 
 (i.e. the pre-boundary field)
 would be unknown.

For convolution transforms, 
 we note that their standard use 
 in constitutive relations is of the time-independent form  $D(t)=\int\epsilon(t-t') E(t') dt'$, 
 but unfortunately, 
 this would only be valid -- obviously -- 
 for time independent media. 
For time \emph{dependent} media, 
 such as one including a time boundary, 
 we would need to use the two-time generalisation $D(t)=\int\epsilon(t,t') E(t') dt'$.

Whilst in principle such generalisations can be written down, 
 it is unfortunate that the time-independent convolution
 does not indicate in any way how a general (two-time) formulation
 could be arrived at, 
 and this is complicated further since the convolution 
 \emph{include} the time boundary itself.
This is in stark contrast 
 to dynamic response models such as \eqref{ABC_DE_P}, 
 or the more general form 
 given here in the appendix as \eqref{ABC_epsilon_sum}, 
 where insisting on (e.g.) continuity of the auxiliary fields is quite natural.

\section{Remarks  --  CR and Ohmic losses}

\vspace{-2mm}
Ohmic losses are also covered by our analysis; 
 as they can also be modelled by a complex permittivity. 
Since the polarisation corresponding to ohmic losses is given
 by $\VP=(\sigma/(-i\omega))\VE$ then this is an alternative dispersive
 constitutive relation,
 which means it does not contradict the statements about constant $\epsilonII$.

\end{document}